\documentclass[11pt,final,DVI12]{scrartcl}
\usepackage{amsmath}
\usepackage{amssymb}
\usepackage{amsthm}
\usepackage{enumitem}
\usepackage[english]{babel}
\usepackage{ifthen}
\usepackage[notref,notcite,color]{showkeys}
\usepackage{multirow}
\usepackage{xspace}
\usepackage{url}
\usepackage[colorlinks,citecolor=blue,final]{hyperref}
\usepackage{tikz}
\usepackage{orcidlink}

\newcommand{\refthm}[1]{Theorem~\ref{#1}\xspace}

\newcommand{\refcon}[1]{Conjecture~\ref{#1}\xspace}

\newcommand{\refsec}[1]{Section~\ref{#1}\xspace}
\newcommand{\refeq}[1]{Equation~\eqref{#1}\xspace}

\newcommand{\itemref}[1]{(\ref{#1})\xspace}

\newcommand{\IFF}{\ensuremath{\mathrel{\Leftrightarrow}}}

\newcommand{\homo}{homomorphism\xspace}

\newcommand{\schuetz}{Sch\"utz\-en\-ber\-ger\xspace}

\newcommand{\CR}{Church-Rosser\xspace}

\newcommand{\swr}{sub\-word-redu\-cing\xspace}
\newcommand{\lr}{length-redu\-cing\xspace}
\newcommand{\wred}{weight-redu\-cing\xspace}

\newcommand{\IRR}{\mathrm{IRR}}

\newcommand{\eex}{\hspace*{\fill}\ensuremath{\Diamond}}

\newcommand\RAS[2]{\overset{#1}{\underset{#2}{\Longrightarrow}}}

\newcommand\ra[1]{\overset{#1}{\longrightarrow}}
\newcommand\LAS[2]{\overset{#1}{\underset{#2}{\Longleftarrow}}}
\newcommand\DAS[2]{\overset{#1}{\underset{#2}{\Longleftrightarrow}}}

\newcommand\RA[1]{\underset{#1}{\Longrightarrow}}
\newcommand\LA[1]{\underset{#1}{\Longleftarrow}}

\newcommand{\set}[2]{\left\{#1\mathrel{\left|\vphantom{#1}\vphantom{#2}\right.}#2\right\}}

\newcommand{\oneset}[1]{\left\{\mathinner{#1}\right\}}
\newcommand{\os}{\oneset}
\newcommand{\sm}{\setminus}
\newcommand{\es}{\emptyset}
\newcommand{\sse}{\subseteq}
\newcommand{\smallset}[1]{\left\{\mathinner{#1}\right\}}

\newcommand{\abs}[1]{\left|\mathinner{#1}\right|}
\newcommand{\Abs}[1]{\left\Vert\mathinner{#1}\right\Vert}

\newcommand{\N}{\ensuremath{\mathbb{N}}}
\newcommand{\Z}{\ensuremath{\mathbb{Z}}}

\newcommand{\PSPACE}{\ensuremath{\mathsf{PSPACE}}}

\renewcommand{\phi}{\varphi}
\newcommand{\eps}{\varepsilon}

\newcommand{\oo}{\omega}

\newcommand{\sig}{\sigma}

\newcommand{\Oh}{\mathcal{O}}

\newcommand{\wh}[1]{\widehat{ #1 }}

\newcommand{\SD}{\mathrm{SD}}

\newcommand{\LTL}{\mathrm{LTL}}
\newcommand{\cLTL}{\mathcal{L\hspace*{-0.1pt}T\hspace*{-2pt}L}}
\newcommand{\FO}{\mathrm{FO}}

\newcommand{\X}{\mathsf{X}}
\newcommand{\F}{\mathsf{F}}

\newcommand{\U}{\mathrel{\mathsf{U}}}

\newcommand{\cU}%
{\mathop{\textcircled{\raisebox{-0.55mm}{\textsf{U}}}}}
\newcommand{\ltrue}{\top}
\newcommand{\lfalse}{\bot}

\newcommand{\coloneq}{\mathrel{\mathord:\mathord=}}

\newtheorem{theorem}{Theorem}[section]
\newtheorem{definition}[theorem]{Definition}
\newtheorem{proposition}[theorem]{Proposition}

\newtheorem{conjecture}[theorem]{Conjecture}
\newtheorem{corollary}[theorem]{Corollary}
\newtheorem{myremark}[theorem]{Remark}
\newtheorem{myexample}[theorem]{Example}

\newenvironment{remark}{\begin{myremark}\normalfont}{\hspace*{\fill}{$\Diamond$}\end{myremark}}
\newenvironment{example}{\begin{myexample}\normalfont}{\hspace*{\fill}{$\Diamond$}\end{myexample}}

\setlist{itemsep=0pt,topsep=4pt}

\renewcommand{\theenumii}{\arabic{enumii}}

\newcommand{\ie}{\textit{i.e.}\xspace}

\newcommand{\wrt}{w.r.t.\xspace}

\newcommand{\mazu}{Mazur\-kie\-wizc\xspace}

\begin{document}

\title{A Survey on the \\ Local Divisor Technique}

\author{Volker Diekert\hspace*{2pt}\orcidlink{0000-0002-5994-3762} \and Manfred Kuf\-leitner\hspace*{2pt}
\orcidlink{0000-0003-3869-416X}}
  

\date{University of Stuttgart, FMI, Germany\\
\texttt{\{diekert,kufleitner\}\@fmi.uni-stuttgart.de}}

\maketitle

\begin{abstract}
  Local divisors allow a powerful induction scheme on the size of a monoid. We survey this technique by giving several examples of this proof method.
  These applications include linear temporal logic, rational expressions with Kleene stars restricted to prefix codes with bounded synchronization delay, Church-Rosser congruential languages, and Simon's Factorization Forest Theorem.
  We also introduce the notion of \emph{localizable language class} as a new abstract concept which unifies some of the proofs for the results above.

The current arXiv-version includes some additional material about 
codes of bounded synchronization delay as well as some updates concerning related literature. 
\end{abstract}

\section{Introduction}\label{sec:intro}

The notion of a \emph{local divisor} refers to a construction 
for finite monoids.  It appeared in this context first 
in \cite{dg06IC} where it was used by the authors as a tool in the proof that local future temporal logic is expressively complete (\wrt~first-order logic) for \mazu traces. The definition 
of a local divisor is very simple: Let $M$ be a finite monoid and $c \in M$. 
Then $cM \cap Mc$ is a semigroup, but it fails to be a submonoid unless 
$c$ is invertible. If  $c$ is not invertible then  $1 \notin cM \cap Mc$ and, as a consequence, 
$\abs{cM \cap Mc} < \abs{M}$. 
The idea is now to turn $cM \cap Mc$ into a monoid by 
defining a new multiplication by $xc \circ cy = xcy$.
This is well-defined and $M_c= (cM \cap Mc, \circ, c)$  becomes a monoid 
where $c$ is the unit. Moreover, if $c$ is not invertible then $M_c$ is a smaller monoid than $M$; and this makes the construction attractive for induction.
(Obviously,  the same idea works for $\os c \cup cMc$ and since $\os c \cup cMc \sse 
cM \cap Mc$ there is a choice here.) 
The original definition for a multiplication of type $xc \circ cy = xcy$ was given for associative algebras. It can be traced back to a technical report of Meyberg, \cite{Mey72}.
He coined the notion of a \emph{local algebra}. Just replace 
$M$ above by a finite dimensional associative algebra (with a unit element) over a field 
$k$. For example, $M$ is the algebra of $n \times n$ matrices over $k$. 
If $c\in M$ is not invertible then the vector space $cM \cap Mc$ has at least one dimension less and $(cM \cap Mc,+, \circ, c)$ is again 
an associative algebra with the unit element $c$. See also \cite{FeTo02} for applications of Meyberg's construction. In formal language theory, the concept of a local divisor is the algebraic mirror of an automata construction used by Thomas Wilke in his Habilitation \cite{Wil98} where parts of of it appeared in \cite{Wil99stacs}, too.

Despite (or more accurately \emph{thanks to}) its simplicity, the \emph{local-divisor technique} is quite powerful.   
For example, it was used in a new and simplified proof for the  Krohn-Rhodes  Theorem \cite{DiekertKS12fi}. 
In \cite{dgk08ijfcs} we applied it to derive results on small fragments of first-order definable languages. In \cite{DiekertKufleitner14tocs} we extended a classical result
of \schuetz from finite words to infinite words by
showing that \emph{$\oo$-rational expressions with bounded synchronization delay} characterize star-free languages. 
At ICALP 2012 we presented a paper which solved a 25 years old conjecture in formal language theory \cite{DiekertKRW15jacm}.
We showed that  regular languages are Church-Rosser congruential.
We come back to this result in more detail in \refsec{sec:crcl}. 
Our result on Church-Rosser languages was obtained in two steps. First, we had to show it for 
regular group languages, which is very difficult and technical. 
This part served as a base for induction. The second part uses 
induction using local divisors. This part is presented here, it is based on the paper \cite{DiekertKW12tcs}.
The construction of local divisors 
has also been an essential  tool in  Kuperberg's work 
on  a linear temporal logic for regular cost functions \cite{Kuperberg14lmcs} which was published 2014.

The journal version of our 2015-arXiv version appeared 2016 in \cite{DiekertK2016tcs}. More recent results are, for example, in
\cite{DartoisGK21,DiekertWalter16,GrayS2022}

\section{Local Divisors}\label{sec:locdiv}
We will apply the local divisor techniques mainly to monoids. However, it is instructive
to place ourselves first in the slightly more general setting of semigroups. Let $S= (S, \cdot)$ be a finite semigroup. 
A  \emph{divisor} $S'$ of $S$  is a homomorphic image of a subsemigroup. 
Let $c \in S$ be any element and consider   $cS \cap Sc$. We can turn the subset $cS \cap Sc$ into a semigroup by defining a new operation $\circ$ as follows:
$$ xc \circ cy = xcy.$$
A direct calculation shows that the operation $\circ$ is well-defined and associative.   
Hence, $S_c= (cS \cap Sc, \circ)$ is a semigroup. 
In order to see that $S_c$ is a divisor consider the 
following subsemigroup $S' = \set{x\in S}{cx \in Sc}$ of $S$. Note that $c \in S'$. 
Define $\phi: S' \to S_c$ by $\phi(x) = cx$. It is surjective since $z \in cS \cap Sc$ implies that we can write $z =cx$ with $x \in S'$. Moreover, $cxy= cx \circ cy$ and 
$S_c$ is the homomorphic image of $S'$. Therefore,  $S_c$ is divisor. We call it the \emph{local 
divisor at $c$}. We want to use $S_c$ for induction. Therefore we characterize next when 
$\abs {S_c} < \abs S$. Recall that $e \in S$ is called an \emph{idempotent} if $e^2 = e$. 
For every finite semigroup there is a natural number $\oo \in \N$ (for example $\oo = \abs S !$) 
such that $x^\oo$ is idempotent for every $x \in S$. 
An element $e$ is called a \emph{unit} if it has a left- and right inverse. In finite semigroups this means that  $e^\oo$ satisfies $e^\oo x = x e^\oo = x$ for all $x \in S$.
Thus, if $S$ contains a unit $e$ then it is monoid and $e^\oo$ becomes the uniquely defined neutral element $1 \in S$.
We have the following result.
 
\begin{proposition}\label{prop:semld}
Let $S$ be a semigroup and $S_c= (cS \cap Sc, \circ)$ be defined as above. 
\begin{enumerate}
\item If $S$ is a monoid, then $S_c= (cS \cap Sc, \circ, c)$ is a monoid and
$S_c$ is divisor in terms of monoids, \ie a homomorphic image of a submonoid $S'$ of $S$. 
\item If $c$ is a unit of $S$, then $S= \set{x\in S}{cx \in Sc}$ and $\phi: S \to S_c$, $x \mapsto cx$ is an isomorphism of monoids.
\item If $c$ is not a unit,  then $\abs {S_c} < \abs S$.
\end{enumerate}
\end{proposition}

\begin{proof}
 (a): Since $S$ is monoid we have $1 \in S'= \set{x\in S}{cx \in Sc}$ and
$S_c$ is the homomorphic image of the submonoid $S'$. 
 
 (b): Trivial. 
 
 (c): If $cS \cap Sc=S$, then we have $cS = S$ and $Sc = S$. This implies that $c$ is a unit. Indeed, we  have $c^\omega S = S = S c^\omega$. For every element $c^\omega x \in S$ we have $c^\omega \cdot c^\omega x = c^\omega x$. Thus, $c^\omega$ is neutral and $c^{\omega-1}$ is the inverse of $c$, \ie, $c$ is a unit. 
Therefore, if $c$ is not a unit then $\abs{S_c} < \abs{S}$.
  \end{proof}

\begin{remark}\label{rem:smallld}
It is clear that $(\os{cc} \cup cSc, \circ)$ is a subsemigroup of $(cS \cap Sc, \circ)$. Moreover, if $S$ is a monoid then 
$(\os c \cup cSc , \circ, c)$ is a submonoid  of $(cS \cap Sc, \circ,c )$. Hence by slight abuse of language, we might call  $(\os{cc} \cup cSc , \circ)$ (resp.\ $(\os c \cup cSc , \circ, c)$) a local divisor 
of $S$, too. Moreover, 
if $c \in S$ is an idempotent, then $(cSc , \circ)= (cSc , \cdot)$ is the usual local monoid at $c$. The advantage is that $\os{cc} \cup cSc$ (resp.{} $\os c \cup cSc$) might be smaller than $cS \cap Sc$. However in worst case estimations there is  no difference. 
\end{remark}

\section{Localizable language classes}\label{sec:loc}

A \emph{language class} $\mathcal{V}$ assigns to every finite alphabet $A$ a set of languages $\mathcal{V}(A^*) \subseteq 2^{A^*}$. A language class $\mathcal{V}$ is \emph{left-localizable} if for all finite alphabets $A$ and $T$ the following properties hold:
\begin{enumerate}
\item\label{aaa:loc} $\emptyset, A^* \in \mathcal{V}(A^*)$.
\item\label{bbb:loc} If $K,L \in \mathcal{V}(A^*)$, then $K \cup L \in \mathcal{V}(A^*)$.
\item\label{ccc:Lloc} For every $c \in A$, the alphabet $B = A \setminus \smallset{c}$ satisfies:
\begin{enumerate}
\item\label{cc1:Lloc} If $K \in \mathcal{V}(B^*)$, then $K \in \mathcal{V}(A^*)$.
\item\label{cc2:Lloc} If $K \in \mathcal{V}(A^*)$ and $L \in \mathcal{V}(B^*)$, then $KcL \in \mathcal{V}(A^*)$.
\item\label{cc3:Lloc} If $K \in \mathcal{V}(B^*)$ and $L \in \mathcal{V}(A^*)$ with $L \subseteq cA^*$, then $KL \in \mathcal{V}(A^*)$.
\item\label{cc4:Lloc} Suppose $g : B^* \to T$ is a mapping
with $g^{-1}(t) \in \mathcal{V}(B^*)$ for all $t \in T$. Moreover,
let $\sig: (cB^*)^* \to T^*$ be defined by 
$\sig(cu_1 \cdots cu_k) =  g(u_1) \cdots g(u_k)$ for 
$u_i \in B^*$. If $K \in \mathcal{V}(T^*)$, then $\sig^{-1}(K) \in \mathcal{V}(A^*)$.
\end{enumerate}
\end{enumerate}
Being \emph{right-localizable} is defined by the right dual
of left-localizability. Properties \itemref{aaa:loc}, \itemref{bbb:loc} and \itemref{cc1:Lloc} are unchanged, but the remaining conditions are replaced by
\begin{enumerate}[start=3]
\item For every $c \in A$, the alphabet $B = A \setminus \smallset{c}$ satisfies:
\renewcommand{\theenumii}{\arabic{enumii}'}
\begin{enumerate}[start=2]
\item\label{cc2:Rloc} If $K \in \mathcal{V}(B^*)$ and $L \in \mathcal{V}(A^*)$, then $KcL \in \mathcal{V}(A^*)$.
\item\label{cc3:Rloc} If $K \in \mathcal{V}(A^*)$ with $K \subseteq A^* c$ and $L \in \mathcal{V}(B^*)$, then $KL \in \mathcal{V}(A^*)$.
\item\label{cc4:Rloc} Suppose $g : B^* \to T$ is a mapping
with $g^{-1}(t) \in \mathcal{V}(B^*)$ for all $t \in T$. Moreover,
let $\sig: (B^*c)^* \to T^*$ be defined by 
$\sig(u_1c \cdots u_kc) =  g(u_1) \cdots g(u_k)$ for 
$u_i \in B^*$. If $K \in \mathcal{V}(T^*)$, then $\sig^{-1}(K) \in \mathcal{V}(A^*)$.
\end{enumerate}
\end{enumerate}
A class of languages $\mathcal{V}$ is \emph{localizable} if it is left-localizable or right-localizable.

\begin{theorem}\label{thm:loc}
If $L \subseteq A^*$ is recognized by a finite aperiodic monoid, then $L \in \mathcal{V}(A^*)$ for every localizable language class $\mathcal{V}$. This means that every localizable language class contains all 
aperiodic languages.
\end{theorem}

\begin{proof}
We can assume that $\mathcal{V}$ be left-localizable; the situation with $\mathcal{V}$ being right-localiza\-ble is symmetric.
Let $h : A^* \to M$ be a homomorphism to a finite aperiodic monoid $M$. It is enough to show $h^{-1}(p) \in \mathcal{V}(A^*)$ for all $p \in M$. We proceed by induction on $(\abs{M},\abs{A})$ with lexicographic order. If $h(A^*) = \smallset{1}$, then either $h^{-1}(p) = \es$ or $h^{-1}(p) = A^*$; 
  and we are done. Hence, we can assume that there is a letter
  $c \in A$ with $h(c) \neq 1$. Let  $B = A \setminus \smallset{c}$ and  $g : B^* \to M$ be the restriction of $h$ to $B^*$. For  all $p \in M$ we have
  \begin{equation}\label{eq:1}
    h^{-1}(p) \;=\;
      g^{-1}(p) \,\cup \!\!\!\!
      \bigcup_{\scriptsize\begin{array}{c}
          p = q r s
        \end{array}} \!\!\!\!
      g^{-1}(q) \cdot
      \big(h^{-1}(r) \cap
      c \/ A^* \cap A^* \hspace*{-0.5pt} c \big) \cdot
      g^{-1}(s)
  \end{equation}
  by factoring every word at the first and the last occurrence of $c$.
  Induction on the size of the
  alphabet yields $g^{-1}(p) \in \mathcal{V}(B^*)$ for all $p \in M$.
  By the closure properties of $\mathcal{V}$, it suffices to show $ \big(h^{-1}(r) \cap
      c \/ A^* \cap A^* \hspace*{-0.5pt} c \big) \cdot g^{-1}(s) \in \mathcal{V}(A^*)$ for every $r \in \varphi(c) M \cap M \varphi(c)$  and $s \in h(B^*)$.
  Let $T = h(B^*)$. 
  In the remainder of this
  proof we will use~$T$ as a finite alphabet.
  The substitution $\sigma: ( c B^* )^* \to T^*$ is defined by 
  \begin{equation*}
    \sigma(c v_1  \cdots c v_k)  =  g(v_1) \cdots g(v_k)
  \end{equation*}
  for $v_i \in B^*$,
  and the homomorphism $f : T^* \to M_c$ to the local divisor $M_c = 
 \big(h(c) M \cap M h(c), \circ, h(c)\big)$ is defined by
  \begin{equation*}
     f\big(g(v)\big) = h(cvc)
  \end{equation*}
  for $v \in B^*$.
  This is well-defined since $h(cvc) = h(c) \hspace*{1pt} g(v) \hspace*{1pt} h(c)$ only depends on $g(v)$ and not on the word $v$ itself.
  Consider a word $w = c v_1 \cdots c v_k$ with $k \geq 0$ and $v_i \in
  B^*$. Then
  \begin{align*}
    f \bigl(\sigma(w) \bigr)
    &= f\bigl( g(v_1) g(v_2) \cdots g(v_k) \bigr) \\
    &= h(c v_1 c) \circ h(c v_2 c) \circ 
    \cdots \circ h(c v_k c) \\
    &= h(c v_1 c v_2 \cdots c v_k c) 
    = h(wc).
  \end{align*} 
  Thus, we have $wc \in h^{-1}(r)$ if and
  only if $w \in \sigma^{-1} \bigl( f^{-1}(r) \bigr)$.  This shows $h^{-1}(r)
  \cap c\/A^* \cap A^* c = \sigma^{-1} \bigl( f^{-1}(r) \bigr) \cdot c$ for every
  $r \in h(c) M \cap M h(c)$.
  It follows that
  \begin{equation*}
    \big(h^{-1}(r) \cap
      c \/ A^* \cap A^* \hspace*{-0.5pt} c \big) \cdot h^{-1}_{c}(s)
  = \sigma^{-1} ( K ) \cdot c \cdot g^{-1}(s)
  \end{equation*}
    for $K = f^{-1}(r)$.
    The monoid~$M_c$ is aperiodic and $\abs{M_c} < \abs{M}$. Induction on the size of the monoid yields $K \in \mathcal{V}(T^*)$, and induction on the alphabet shows $g^{-1}(t) \in \mathcal{V}(B^*)$ for all $t \in T$. By the closure properties of $\mathcal{V}$ we obtain $\sigma^{-1}(K) \in \mathcal{V}(A^*)$ and $\sigma^{-1}(K) \cdot c \cdot g^{-1}(s) \in \mathcal{V}(A^*)$.
   This concludes the proof.
\end{proof}

\section{Linear temporal logic}\label{sec:mkltl}

By Kamp's famous theorem~\cite{kam68}, \emph{linear temporal logic} LTL over 
words has the same expressive power as first-order logic $\FO[<]$. In an algebraic setting, 
one shows first that every first-order definable language $L\sse A^*$ 
is aperiodic. This is relatively easy and no local divisor technique applies here.  In this section we concentrate on the converse. We give a simple proof 
that every aperiodic language is LTL-definable. 
 We give the proof for finite words, only. 
However, the basic proof techniques generalize to infinite words~\cite{dg08SIWT} and also to Mazurkiewicz traces~\cite{dg02jcss}.

The syntax of \emph{linear temporal logic $\LTL(A)$} over an alphabet $A$ is defined as follows:
\begin{equation*}
  \varphi ::= \ltrue \mid a \mid \neg \varphi \mid (\varphi \vee \varphi) \mid \X \varphi \mid (\varphi \U \varphi)
\end{equation*}
for $a \in A$. The modality $\X$ is for ``ne$\X$t'' and $\U$ is for ``$\U$ntil''. As usual, we omit the bracketing whenever there is no confusion.
For the semantics we interprete a word $u = a_1 \cdots a_n$ with $a_i \in A$ as a labeled linear order with positions $\oneset{1,\ldots,n}$, and position $i$ is labeled by $a_i$. We write $u,i \models \varphi$ if the word $u$ at position $i$ models $\varphi$, and we write $u,i \not\models \varphi$ if this is not the case. The semantics of $\LTL(A)$ is defined by:
\begin{alignat*}{2}
  u,i &\models \ltrue &\qquad&\text{is always true} \\
  u,i &\models a &\quad\IFF\quad& a_i = a \\
  u,i &\models \neg \varphi &\quad\IFF\quad& u,i \not\models \varphi \\
  u,i &\models \varphi \vee \psi &\quad\IFF\quad& u,i \models \varphi \,\text{ or }\, u,i \models \psi \\
  u,i &\models \X \varphi &\quad\IFF\quad& i<n \,\text{ and }\, u,i+1 \models \varphi \\
  u,i &\models \varphi \U \psi &\quad\IFF\quad& \text{there exists $k \in \oneset{i,\ldots,n}$ such that }\, u,k \models \psi \\[-1mm]
  &&&\text{and for all $j \in \oneset{i,\ldots,k-1}$ we have }\, u,j \models \varphi
\end{alignat*}
The formula $\varphi \U \psi$ holds at position~$i$ if there exists a position~$k\geq i$ such that~$\psi$ holds at~$k$ and all positions from $i$ to $k-1$ satisfy $\varphi$. 
A formula $\varphi$ in $\LTL(A)$ defines the language 
\begin{equation*}
  L(\varphi) = \set{u \in A^+}{u,1 \models \varphi}.
\end{equation*}
This means that when no position is given, then we start at the first position of a nonempty word.
We introduce the following macros:
\begin{align*}
  \lfalse &\ \coloneq \ \neg \ltrue 
  & B &\ \coloneq \ {\textstyle \bigvee_{b\in B} b} \qquad \text{for $B \subseteq A$} \\
  \varphi \wedge \psi &\ \coloneq \ \neg (\neg \varphi \vee \neg \psi)
  & \F \varphi &\ \coloneq \ \ltrue \U \varphi
\end{align*}
The macro $\F \varphi$ (for ``$\F$uture'') holds at position $i$ if $\varphi$ holds at some position $k \geq i$.
For $L,K \sse A^*$ we define a variant of the Until-modality on languages by
\begin{equation*}
  K \cU L = \set{vw \in A^*}{\text{$w \in L$, $\forall \, v=pq$ with $q \neq \varepsilon$: $qw \in K$}}.
\end{equation*}
The language class $\cLTL$ resembles the behaviour of $\LTL$ by using a more global semantics. The languages in $\cLTL(A^*)$ are inductively defined by:
\begin{itemize}
\item $\emptyset \in \cLTL(A^*)$.
\item If $K,L \in \cLTL(A^*)$ and $a \in A$, then $A^* \setminus L, K \cup L, aL, K \cU L \in \cLTL(A^*)$.
\end{itemize}
The formal connection between $\cLTL$ and $\LTL$ is given by the following proposition.

\begin{proposition}\label{prp:LTLcLTL}
  We have $L \in \cLTL(A^*)$ if and only if $L \setminus \smallset{\varepsilon}$ is definable in $\LTL(A)$.
\end{proposition}

\begin{proof}
  We first show $L(\varphi) \in \cLTL(A^*)$ for every formula $\varphi \in \LTL(A)$. We have $A^* = A^* \setminus \emptyset \in \cLTL(A^*)$. For $\varphi \coloneq \ltrue$ we have $L(\ltrue) = A^+ = \bigcup_{a \in A} aA^* \in \cLTL(A^*)$. For $\varphi \coloneq a$ we have $L(a) = aA^* \in \cLTL(A^*)$. The construction for negations is $L(\neg \psi) = L(\ltrue) \setminus L(\psi)$, and disjunctions translate into unions. If $\varphi \coloneq \X \psi$, then $L(\X \psi) = \bigcup_{a \in A} a L(\psi)$. Finally, if $\varphi \coloneq \psi_1 \U \psi_2$, then $L(\varphi) = L(\psi_1) \cU L(\psi_2)$.
  
  For the converse, we show that for every language $L \in \cLTL(A^*)$ there exists a formula $\varphi_L \in \LTL(A)$ such that $L(\varphi_L) = L \setminus \smallset{\varepsilon}$. If $L = \emptyset$, then $\varphi_\emptyset = \lfalse$. Complements translate into negations, and unions translate into disjunctions. If $\varepsilon \not\in K$, then the formula for $L = aK$ is $\varphi_{aK} \coloneq a \wedge \X \varphi_K$. If $\varepsilon \in K$, then the formula of $L = aK$ is $\varphi_{aK} \coloneq a \wedge (\neg\X\ltrue \vee \varphi_{K} )$. If $L = K_1 \cU K_2$ for $\varepsilon \not\in K_2$, then $\varphi_L \coloneq \varphi_{K_1} \U \varphi_{K_2}$. Finally, if $L = K_1 \cU K_2$ for $\varepsilon \in K_2$, then $\varphi_L \coloneq (\varphi_{K_1} \U \varphi_{K_2}) \vee \neg \F \neg \varphi_{K_1}$; the formula $\neg \F \neg \varphi_{K_1}$ says that all positions satisfy $\varphi_{K_1}$.
\end{proof}

\begin{proposition}\label{prp:cLTLloc}
  The language class $\cLTL$ is left-localizable.
\end{proposition}

\begin{proof}
  The properties \itemref{aaa:loc} and \itemref{bbb:loc} are obvious. Let now $c \in A$ and $B = A \setminus \smallset{c}$. The language $B^+$ is defined by $\neg \F \neg B \in \LTL(A)$ and thus $B^+ \in  \cLTL(A^*)$. Since $\smallset{\varepsilon} = A^* \setminus \big( \bigcup_{a \in A} A^* \big)$, we have $B^* \in \cLTL(A^*)$.
  
  For \itemref{cc1:Lloc} let $K,L \in \cLTL(B^*)$. By induction we can assume $K,L \in \cLTL(A^*)$. This immediately yields $K \cup L, aL, K \cU L \in \cLTL(A^*)$ for all letters $a$. The set $B^* \setminus L$ can be written as $(A^* \setminus L) \cap B^*$ and hence $B^* \setminus L$ is in $\cLTL(A^*)$.
This shows $\cLTL(B^*) \subseteq \cLTL(A^*)$.

  For \itemref{cc2:Lloc} let $K \in \cLTL(A^*)$ and $L \in \cLTL(B^*)$. We have $K c L = A^* c L \cap K c B^*$ because the last $c$ in word is unique. Note that $A^* c L = A^* \cU cL \in \cLTL(A^*)$. It remains to show $KcB^* \in \cLTL(A^*)$ by structural induction: 
  \begin{align*}
  	(A^* \setminus L') cB^* &= A^* c B^* \setminus L' c B^*
	\\ (K' \cup L') c B^* &= K' a B^* \cup L' a B^*
	\\ (aL') c B^* &= a (L'cB*)
	\\ (K' \cU  L') c B^* &= (K'cB^*) \cU (L'cB^*).
  \end{align*} 

  For \itemref{cc3:Lloc} let $K \in \cLTL(B^*)$ and $L \in \cLTL(A^*)$ with $L \subseteq cA^*$. We have $KL = B^* L \cap K c A^*$. Note that $B^* L = BA^* \cU L \in \cLTL(A^*)$. As before, one can easily show $K c A^* \in \cLTL(A^*)$ by structural induction. For instance, $(B^* \setminus L') c A^* = B^* c A^* \setminus L' c A^*$ since $L' \subseteq B^*$ and the occurrence of the first $c$ is unique.

  For \itemref{cc4:Lloc} suppose $g : B^* \to T$ is a mapping
with $g^{-1}(t) \in \cLTL(B^*)$ for all $t \in T$. Moreover,
let $\sig: (cB^*)^* \to T^*$ be defined by 
$\sig(cu_1 \cdots cu_k) =  g(u_1) \cdots g(u_k)$ for 
$u_i \in B^*$. We show $\sig^{-1}(K) \in \cLTL(A^*)$ for every $K \in \cLTL(T^*)$ by structural induction on~$K$. For all $K,L \subseteq T^*$ and $t \in T$ we have:
  \begin{align*}
    \sig^{-1}(T^*) &= \smallset{\varepsilon} \cup cA^*
    \\ \sig^{-1}(K \setminus L) &= \sig^{-1}(K) \setminus \sig^{-1}(L)
    \\ \sig^{-1}(K \cup L) &= \sig^{-1}(K) \cup \sig^{-1}(L)
    \\ \sig^{-1}(t L) &= c \cdot g^{-1}(t) \cdot \sig^{-1}(L)
    \\ \sig^{-1}(K \cU L) &= \big((\sig^{-1}(K) \cup BA^*)\U \sig^{-1}(L)\big) \cap \sig^{-1}(T^*).
  \end{align*}
  Note that $g^{-1}(t) \cdot \sig^{-1}(L) \in \cLTL(A^*)$ by \itemref{cc3:Lloc}.
\end{proof}

Together with Theorem~\ref{thm:loc} this leads to the following result.

\begin{corollary}\label{cor:dennichbineincor}
  If $L \subseteq A^+$ is recognized by a finite aperiodic semigroup, then $L$ is definable in $\LTL(A)$.
\end{corollary}

\begin{proof}
  As a subset of $A^*$, the language $L$ is recognized by a finite aperiodic monoid. By Theorem~\ref{thm:loc} and Proposition~\ref{prp:cLTLloc} we have $L \in \cLTL(A^*)$. Since $\varepsilon \not\in L$, Proposition~\ref{prp:LTLcLTL} shows that $L$ is definable in $\LTL(A)$.
\end{proof}

\section{Bounded synchronization delay}\label{sec:sd}
A fundamental and classical result of \schuetz, published in~1965,
says that a language is star-free languages if and only if its syntactic monoid is finite and aperiodic~\cite{sch65sf}. 
There is another beautiful, but less known,  characterization of the star-free languages due to
\schuetz~\cite{schue75}, which seems to be quite overlooked. It
characterizes the star-free languages without using complementation, but the
inductive definition allows the star-operation on languages $K$
(already belonging to the class) if $K$ is a prefix code with bounded
synchronization delay. Since synchronization delay is the main feature
in this approach, the class is denoted by $\SD$.  The notion of
bounded synchronization delay was introduced by Golomb et al.~\cite{GolombG65ic,GolombGW58}. It is an important concept in coding theory.

A language $K \sse A^*$ is called \emph{prefix-free} if $u, uv \in K$
implies $u=uv$. A prefix-free language $K \sse A^+$ 
is also called a \emph{prefix code} since every word $u \in K^*$
admits a unique factorization $ u = u_1 \cdots u_k$ with $k \geq 0$
and $u_i \in K$. 
A prefix code $K$ has \emph{synchronization delay $d$} if for all $u,v, w \in A^*$ we have:
\begin{equation*}
  \text{if \,} uvw\in K^* \,\text{ and }\, v \in K^d 
  \text{, \,then }\, uv \in K^*.
\end{equation*}
Note that $uv \in K^*$ and $uvw \in K^*$ implies $w \in K^*$ since $K$ is a prefix code.
The prefix code $K$ has \emph{bounded synchronization delay} if there is some $d \in \N$ such that $K$ has synchronization delay $d$. Note that every subset
$B \sse A$ yields a prefix code
with synchronization delay $0$.  In particular, the sets $B$ are
prefix codes of bounded synchronization delay for all $B \sse A$.

The intuition behind this concept is the following: Assume a sender
emits a stream of code words from $K$, where $K$ is a prefix code with synchronization delay~$d$. If a receiver misses the beginning
of the message, he can wait until he detects a sequence of $d$ code
words. Then he can synchronize and decipher the remaining text after
these $d$ words.

The prefix code $K = \os{a,ba^d}$ over the two-letter alphabet $A = \os{a,b}$ has minimal synchronization delay $d$. To see that it does not have synchronization delay $d-1$, let $u=b$, $v=a^{d-1} \in K^{d-1}$, $w = a$. Then $uvw = ba^d \in K^*$ but $uv = ba^{d-1} \not\in K^*$. To see that it has synchronization delay $d$, let $uvw \in K^*$ and $v \in K^d$. If $v$ contains the letter $b$, then this occurrence corresponds to the code word $ba^d$ both in $v$ and $uvw$. Otherwise, $v = a^d$ and we have $uv \in K^*$ because every $b$ in $uv$ is followed by the factor $a^d$.

We now inductively define \schuetz{}'s language class $\SD$:
\begin{enumerate}
\item We have $\es \in \SD(A^*)$ and $\smallset{a} \in
  \SD(A^*)$ for all letters $a \in A$.
\item If $K,L \in \SD(A^*)$, then\; $K \cup L,\, K\cdot L\;\in
  \SD(A^*)$.
\item\label{sd:star} If $K \in \SD(A^*)$ is a prefix code 
  with bounded synchronization delay, then $K^* \in
  \SD(A^*)$.
\end{enumerate}
Note that, unlike the definition of star-free sets, the inductive
definition of $\SD(A^*)$ does not use any complementation. According to our definition, the empty set is a prefix code of delay $0$. Hence, $\es^*=\os{\eps} \in A^*$. 
The alternative is to replace the first item above, for example, by requiring that all finite subsets of $A^*$ belong to $\SD(A^*)$. 

\begin{proposition}\label{prp:SDloc}
  The language class $\SD$ is right-localizable.
\end{proposition}

\begin{proof}
  The properties \itemref{aaa:loc}, \itemref{bbb:loc}, \itemref{cc1:Lloc}, \itemref{cc2:Rloc}, and \itemref{cc3:Rloc}  are obvious. Let now $c \in A$ and $B = A \setminus \smallset{c}$ and consider property \itemref{cc4:Rloc}.

  Suppose $g : B^* \to T$ is a mapping
with $g^{-1}(t) \in \SD(B^*)$ for all $t \in T$. Moreover,
let $\sig: (B^*c)^* \to T^*$ be defined by 
$\sig(u_1 c \cdots u_k c) =  g(u_1) \cdots g(u_k)$ for 
$u_i \in B^*$. We show $\sig^{-1}(K) \in \SD(A^*)$ for every $K \in \SD(T^*)$ by structural induction on~$K$:
\begin{align*}
  \sigma^{-1}(t) &= g^{-1}(t) c
  \\ \sigma^{-1}(K \cup L) &= \sigma^{-1}(K) \cup \sigma^{-1}(L)
  \\ \sigma^{-1}(K \cdot L) &= \sigma^{-1}(K) \cdot \sigma^{-1}(L)
  \\ \sigma^{-1}(K^*) &= \sigma^{-1}(K)^*\,.
\end{align*}
It remains to verify 
that $\sigma^{-1}(K)$ is a prefix code of bounded
synchronization delay whenever $K$ has this property.  Clearly, $\varepsilon \notin
\sigma^{-1}(K)$.  To see prefix-freeness, consider $u,uv \in
\sigma^{-1}(K)$.  This implies $u \in A^*c$ and hence, $\sigma(uv) =
\sigma(u) \sigma(v)$. It follows that $\sig(v) = \varepsilon$ because $K$ is
prefix-free. This implies $v = \varepsilon$ since the length of $\sigma(v)$ is the number of $c$'s in $v \in (B^*c)^*$.

Finally, let $L = \sig^{-1}(K)$ and suppose $K$ has
synchronization delay $d$. We show that~$L$ has synchronization delay
$d+1$: Let $uvw \in L^*$ with $v \in L^{d+1}$. Write $v = u'cv'$ with
$v' \in L^d$. Note that $v' \in A^* c$.  It follows $\sigma(uv) =
\sigma(uu'c) \sigma(v')$ and $\sigma(v') \in K^d$. Thus, $\sigma(uv)
\in K^*$. We obtain $uv \in L^*$ as desired.
\end{proof}

\begin{corollary}\label{cor:APinSD}
  If $L \subseteq A^*$ is recognized by a finite aperiodic semigroup, then $L \in \SD(A^*)$.
\end{corollary}

\begin{proof}
  This is an immediate consequence of Proposition~\ref{prp:SDloc} and Theorem~\ref{thm:loc}.
\end{proof}

A language is \emph{star-free} it can be built from the finite languages by using concatenation and Boolean operations. One can think of the star-free languages as rational languages where the Kleene-star is replaced by complemention.

\begin{theorem}\label{thm:APSDSF}
  If $L \subseteq A^*$. Then the following statements are equivalent.
  \begin{enumerate}
\item\label{aaa:APSDSF} $L$ is recognized by a finite aperiodic semigroup.
\item\label{bbb:APSDSF} $L \in \SD(A^*)$.
\item\label{ccc:APSDSF} $L$ is star-free.
\end{enumerate}
\end{theorem}

\begin{proof}
\itemref{aaa:APSDSF}~implies~\itemref{bbb:APSDSF}: This is an immediate consequence of Proposition~\ref{prp:SDloc} and Theorem~\ref{thm:loc}.

\itemref{bbb:APSDSF}~implies~\itemref{ccc:APSDSF}:
It suffices to shows that $K^*$ is star-free if $K$ is a star-free prefix code with bounded synchronization delay.
As $K$ is a prefix code, we can write $A^*\setminus K^*$ as
  an infinite union:
  \begin{equation}\label{eq:code}
    A^*\setminus K^* = \bigcup_{i \geq 0} 
    \left( K^iA A^* \sm  K^{i+1} A^* \right).
  \end{equation}
  Now, let $d$ be the synchronization delay of $K$. Then we can write
  \begin{equation*}
    A^* \setminus K^* = A^* K^d(A A^* \sm  K A^*) 
    \cup\bigcup_{0\leq i < d} (K^iA A^*\sm  K^{i+1} A^*).
  \end{equation*}
  The inclusion from left to right follows from \refeq{eq:code}. The
  other inclusion holds since the intersection of $K^*$ and $A^*
  K^d(A A^* \sm K A^*)$ is empty. This is obtained by using the
  definition of synchronization delay.
  
\itemref{ccc:APSDSF}~implies~\itemref{aaa:APSDSF}: This is verified by showing that the syntactic monoid of every star-free language is aperiodic, see e.g.~\cite{pin86} for definitions and basic properties of syntactic monoids. 
 
 We prove the following claim.
 For every language star-free language $K$ there exists an integer $n(K)
  \in \N$ such that for all words $p,q,u,v \in A^*$ we have
  \begin{equation*}
    p\, u^{n(K)}q \in K \ \Leftrightarrow \ p\, u^{n(K)+1}q \in K.
  \end{equation*}

  For the languages $A^*$ and $\smallset{a}$ with $a \in A$ we define $n(A^*) = 0$ and $n(\smallset{a}) = 2$.
  Let now $K,K'$ be star-free such that $n(K)$ and $n(K')$ exist.
  We set
  \begin{gather*}
    n(K \cup K') =  n(K \setminus K') = \max \bigl( n(K), n(K') \bigr), \\
    n(K \cdot K') = n(K) + n(K') + 1.
  \end{gather*}
  The correctness of the first two choices is straightforward. For the last equation,
  suppose $p\, u^{n(K)+n(K')+2}q \in K \cdot K'$. Then either $p\,
  u^{n(K)+1} q' \in K$ for some prefix $q'$ of $u^{n(K')+1} q$ or
  $p'\, u^{n(K')+1} q \in K'$ for some suffix $p'$ of $p u^{n(K) +
    1}$. By definition of $n(K)$ and $n(K')$ we have $p\, u^{n(K)} q' \in
  K$ or $p'\, u^{n(K')} q \in K'$, respectively. Thus $p\,
  u^{n(K)+n(K')+1}q \in K \cdot K'$. The other direction is similar:
  If $p\, u^{n(K)+n(K')+1}q \in K \cdot K'$, then $p\,
  u^{n(K)+n(K')+2}q \in K \cdot K'$.
\end{proof}

\schuetz~\cite{Schutzenberger1974b} generalized the concept of bounded synchronization further beyond aperiodic languages. In \cite{DiekertWalter16}, this was generalized even further and led to the notion of a \emph{group-controlled star-operator}. Another extension of \cite{sch65sf} 
is in \cite{PlaceZeitounICALP2019}. 

\subsection*{Block codes}
To give some background on the concept of bounded synchronization delay, 
we conclude this section by some examples and remarks on block codes.
A \emph{block code} is a nonempty code in which all code words have the same length. In particular, it is a prefix code. 
Block codes,
where no concatenation of two code words contains a proper infix of a code word, have been studied for more than been sixty years \cite{CrickGO1957}, and they are traditionally called \emph{comma-free}. Formally, a (possibly variable-length code $K\sse A^*$ is comma-free if for all $u,v\in K$ and $p,q,w\in A^+$ the equation $uv=pwq$ implies $w\notin K$. In our notation, a comma-free code is a code with bounded synchronization delay~$1$.

The original motivation  was related to biology, and the connection to genetics is still of interest, see for example the preprint 
\cite{FimmelMPSS2019} and the references therein.

The following example shows that there are block codes with arbitrarily high synchronization delay. 

\begin{example}\label{ex:fred}
Let
\begin{align*}
  A &= \os{a_1,\ldots,a_{d}} \cup \os{b_1, \ldots, b_{d}},
  \\ K &= \set{a_i b_{i}}{1 \leq i \leq d} \cup \set{b_i a_{i+1}}{1 \leq i < d}.
\end{align*}
All words in $K$ have length $2$. We first show that $K$ does not have synchronization delay $d-1$: Let $u = a_1$, $v= (b_1 a_2)(b_2 a_3) \cdots (b_{d-1} a_{d}) \in K^{d-1}$, and $w=b_d$. Then $uvw = (a_1 b_1) \cdots (a_d b_d) \in K^*$ but $uv \not\in K^*$.

Next, we show that $K$ has synchronization delay $d$: Suppose that $uvw \in K^*$ and $v \in K^d$.
If $b_d$ occurs in $v$, then we can synchronize at the factor $a_d b_d$ because there is only one word in $K$ which contains the letter $b_d$.

Otherwise, some letter $b_i$ occurs at least twice in $v$. In particular, there exists a factor $b_i v' b_j$ with $i \geq j$ in $v$. We can choose this factor such that the length of $v'$ is minimal. In particular, there is no letter $b_k$ in $v'$. If $b_i v'$ ends with $a_j$, then it cannot end with $b_{j-1} a_j$ because $i > j-1$. Therefore, this occurrence of $a_j$ corresponds to the code word $a_j b_j \in K$ both in $v$ and $uvw$. 

If $v'$ does not end with $a_j$, then this occurrence of $b_j$ (in the factor  $b_i v' b_j$) needs to be followed by $a_{j+1}$ in $v$. In particular, we can synchronize at $b_j a_{j+1}$.

This shows that $d$ is the minimal synchronization delay of $K$.

It is easy to generalize this example to longer block codes by adding two letters $c,d$ to the alphabet $A$. For $k\geq 2$ and $W = \os{c,d}^{k-2}$ the code
\begin{align*}
  K' &= \set{a_i b_{i} W}{1 \leq i \leq d} \cup \set{b_i W a_{i+1}}{1 \leq i < d}
\end{align*}
has minimal synchronization delay $d$ and all words in $K'$ have length $k$.
\end{example}

\begin{remark}
  If a block code $K \subseteq A^k$ has bounded synchronization delay, then it only contains primitive words and there cannot be any two conjugated words in $K$. Golomb and Gorden have shown that choosing the lexicographically minimal word in every conjugacy class of primitive words yields a block code with bounded synchronization delay~\cite{GolombG65ic}; in particular, this achieves the maximal number of words in the block code $K$ (for fixed alphabet $A$ and length $k$) such that $K$ has bounded synchronization delay.
Example~\ref{ex:fred} with the order $a_1 < b_1 < a_2 < b_2 < \ldots < a_d < b_d < c < d$ yields a lower bound of $d$ on the minimal synchronization delay of any such block code over at least $2d+1$ letters (and block length at least $2$; note that we could use any non-empty alphabet instead of $\os{c,d}$). On the other hand, Eastman~\cite{Eastman1965ieee} showed that for every odd number $k$, we can choose one word in every conjugacy class of primitive words in $A^k$ such that the resulting block code is comma-free; in other words, for odd $k$, synchronization delay $1$ can be achieved with the maximal number of words.
  
  For $A = \os{a,b}$ and $k=5$, the construction of Golomb and Gordon yields $K = \os{aaaab, aaabb, aabab, aabbb, ababb, abbbb}$. This code does not have synchronization delay $1$, e.g., witnessed by $u=aa$, $v=aabab \in K$, $w = bbb$. On the other hand, the code $K' = \os{abaaa, ababa, ababb, abbaa, abbba, abbbb}$ given by Scholtz~\cite{Scholtz1969ieee} has synchronization delay $1$.
\end{remark}

\section{Church-Rosser congruential languages}\label{sec:crcl}
  
\emph{Word Problem} WP$(L)$ of a language $L \subseteq A^*$ is following computational task. 
\begin{center}
 \textbf{ Input:} $ w \in A^*$.  \textbf{ Question:} Do we have $ w \in L$?
\end{center}
The following facts are standard in formal language theory. 
 \begin{itemize}
\item If $L$ is regular, then  WP$(L)$ is decidable in  real time.
\item If $L$ is deterministic context-free then WP$(L)$ is solvable in linear time.
\item If $L$ is context-free then WP$(L)$ is solvable in less than cubic time.
\item If $L$ is context-sensitive then WP$(L)$ is solvable in polynomial space, and there are context-sensitive languages such that 
WP$(L)$ is \PSPACE{} complete.
\end{itemize}
The paper of McNaughton, Narendran, and Otto \cite{McNaughtonNO88} exploits  the  following theme:  ``Go beyond  deterministic context-free and keep linear time soluability for the word problem by using Church-Rosser semi-Thue systems.''

Before we proceed we need more preliminaries and notation. An element of the alphabet $A$ is
called a {\em letter}. A {\em word} is an element of  $A^*$. The empty
word is denoted by 1. The {\em length} of a word $u$ is denoted by
$\abs u$. We have $\abs u = n$ for $u= a_1 \cdots a_n$ where $a_i \in
A$. The empty word has length $0$.  We carefully distinguish between
the notion of factor and subword.  Let $u,v \in A^*$. The word $u$ is
called a {\em factor} of $v$ if there is a factorization $v= xuy$. It
is called a {\em subword} of $v$ if there is a factorization $v=
x_0u_1x_1 \cdots u_kx_k$ such that $u= u_1 \cdots u_k$.  A subword is
also sometimes called a {\em scattered subword} in the literature.
A \emph{weight} means here a \homo
$\Abs{\cdot}: A^* \to \N$ such that $\Abs a  >0 $ for all letters $a \in A$. 
 The length function is a weight. If the weight $\Abs{\cdot}$ is given we say that $(A,\Abs{\cdot})$ is a \emph{weighted alphabet}. 
 
A {\em semi-Thue system} over $A$ is a subset $S\sse A^*\times
A^*$. The elements are called {\em rules}. We frequently write
$\ell \ra{} r$ for rules $(\ell,r)$.  A system $S$ is called
{\em \lr} (resp.{}\emph{ \wred} for a weight $\Abs{\cdot}$) if we have $\abs \ell > \abs r $ 
(resp.{}$\Abs\ell > \Abs r $) for all rules $(\ell,r)
\in S$. It is called {\em \swr,} if $r$ is a subword of $\ell$ and
$\ell \neq r$ for all rules $(\ell,r) \in S$. Every \swr system is
\lr and \wred for all weights. 
Every system $S$ defines the rewriting relation ${\RA S} \sse A^*
\times A^*$ by
\begin{align*}
  u \RA S v \;\text{ if } \; u=p\ell q, \; v= prq \; \text{ for some
    rule } \; (\ell,r) \in S.
\end{align*}

By $\RAS*{S}$ we mean the reflexive and transitive closure of
$\RA{S}$.  By $\DAS*{S}$ we mean the symmetric, reflexive, and
transitive closure of $\RA{S}$. We also write $u \LAS*{S}v$ whenever
$v\RAS*{S}u$.  The system $S$ is {\em confluent} if for all
$u\DAS*{S}v$ there is some $w$ such that $u\RAS*{S}w\LAS*{S}v$.
By $\IRR(S)$ we denote the set of irreducible words, \ie the set of
words where no left-hand side occurs as any factor.  The relation
${\DAS * S} \sse A^* \times A^*$ is a congruence, hence the congruence
classes $[u]_S = \{v \in A^*\mid u \DAS*S v\}$ form a monoid which is
denoted by $A^*/\DAS*S$ or simply by  $A^*/S$. 

\begin{definition} 
  A semi-Thue system $S\sse A^* \times A^*$ is called a {\em Church-Rosser system} if it
  is \lr and confluent.  A language $L \sse A^*$ is called a {\em
    Church-Rosser congruential language} if there is a finite
  Church-Rosser system $S$ such that $L$ can be written as a finite
  union of congruence classes $[u]_S$. If, in addition,
  $A^*/S$ itself is a finite monoid then $L \sse A^*$ is called a {\em
    strongly Church-Rosser congruential}. If $A^*/S$ is finite then 
    we say that $S$ is of  \emph{finite index}. 
 \end{definition}

The motivation to consider these languages in \cite{McNaughtonNO88} stems from the following. 

\begin{remark}\label{rem:basic}
  Let $S\sse A^* \times A^*$ be a \wred system. Then on input 
  $w \in A^*$ of length $n$ we can compute in time $\Oh(n)$ 
  some word $\wh w \in \IRR(S)$ such that $w \RAS*S \wh w$. 
  In particular, if $L$ is a Church-Rosser congruential language
  then its word problem is solvable in linear time. 
\end{remark}
 
 Let us consider some examples. 
 \begin{itemize}
 \item Let $S= \os{aab \to ba, cb \to c}$. It is \CR, and hence 
 $L_0= [ca]_S$ is \CR congruential.  The language $L_0$ is not 
 context-free since $L_0 \cap ca^*b^* = \set{ca^{2^n}b^n}{n \geq 0}$. 
 Therefore the class of \CR congruential languages is not included in the class of context-free languages. 
  \item Let $L_1 = \set{a^nb^n}{n \geq 0}$. It is Church-Rosser congruential due to $S = \oneset{aabb \to ab}$ and 
 $L_1 = [ab]_{S} \cup
      [{1}]_{S}$. The monoid 
 $A^* / {S}$ is infinite because $L$ is not regular. 
 We may also note that  $[a^n]_{S} = \os {a^n}$ for $n \geq 1$.
 Hence there are infinitely many classes. 
   \item Let $L_2 = \set{a^mb^n}{m \geq n \geq 0}$, it 
    is  deterministic context-free, but not Church-Rosser congruential since
      $a^m$ must be irreducible for each  $m \geq 1$.
           \item Let $L_3 = \smallset{a,b}^* a \smallset{a,b}^*$.
    It is strongly Church-Rosser congruential due to 
    $S = \oneset{aa \to a,\ b \to {1}}$, $L_3 = [a]_{S}$, and 
    $A^* / {S} = \os{[1]_{S}, [a]_{S}}$. 
     \item Let $L_4 = (ab)^*$ and 
    \ $S = \oneset{aba \to a}$. The system $S$ is \CR and 
   $L_4 = [ab]_{S} \cup [{1}]_{S}$. However, 
   $A^* / {S}$ is infinite although $L_4$ is regular. Therefore $S$ does not show that $L_4$ is strongly \CR congruential. 
   However, choosing  
   $T = \{aaa \to aa,\ aab \to aa,\ baa \to aa,
      bbb \to aa,\ bba \to aa,\ abb \to aa, aba \to a,\ bab \to b\}$, we obtain  
 $L_4 = [ab]_{T} \cup [{1}]_{T}$ and  $A^* / {T}$ has $7$ elements, only.   Hence, $L_4$ is indeed strongly \CR congruential. \eex 
\end{itemize}

The languages $L_0$ and $L_2$ show that the 
      classes of (deterministic) context-free languages and 
      \CR congruential languages are incomparable. Therefore 
      in \cite{McNaughtonNO88} a weaker  notion of \CR  languages
has been considered, too. The new class contained all \CR congruential languages as well as all deterministic context-free languages, still their word problems are solvable in linear time using \CR semi-Thue systems. We do not go into details, but focus on the following  conjecture dating back to 1988 and its solution in 2012.

\begin{conjecture}[\cite{McNaughtonNO88}]\label{con:MNO88}
Every regular language is \CR congruential.
\end{conjecture}
After some significant initial
progress on this conjecture in~\cite{Narendran84phd,NiemannPhD02,NiemannO05,NiemannW02,reinhardtT03}
there was stagnation.
 Before 2011 the most advanced result was 
 a full solution when the monoid $M$ equals $\Z/2 \Z$ \cite{NiemannW02}. More general, it was  announced in 2003 by
Reinhardt and Th\'erien in \cite{reinhardtT03} that 
\refcon{con:MNO88} is true for all regular languages where the
syntactic monoid is a group. However, the manuscript has never been
published as a refereed paper and there are some flaws in its
presentation.
 
Let us continue with some examples which show that 
this statement is far from being trivial even for finite cyclic groups. It shows that a main difficulty is in the number of generators, because the  
corresponding the monoid $M$ is always equal the cyclic group $\Z / 3 \Z$ which is the next step beyond $\Z/2 \Z$.
 \begin{itemize}
\item Let $L_5 = \set{w \in a^*}{\abs{w} \equiv 0 \bmod 3}$, 
then $S = \smallset{aaa \to {1}}$ shows that $L_5$ is strongly \CR congruential.
  \item Let $L_6 = \set{w \in \smallset{a,b}^*}{\abs{w} \equiv 0 \bmod 3}$. We have $L_6 = [1]_S$ \wrt~to the system
    $S = \set{u \to {1}}{\abs{u} = 3}$. But 
    $S$ is not confluent, as we can see from $a
      \;{\LA{S}}\; aabb \;{\RA{S}}\; b$. The smallest system (we are aware of) 
      showing that $L_6$ is \CR congruential is rather large. 
      We may choose 
$T = \{aaa \to 1,\ baab \to b,\  (ba)^3b \to b\}\cup\;\{bb\,u\,bb \to b^{\abs{u}+1}\mid 1 \leq
        \abs{u} \leq 3\}$. The language $L_6$ is a union of elements in $A^* / T$ and $A^*/ T$ contains $272$ elements,  
      longest irreducible word has length $16$.
   \item Finally, consider $L_7 = \set{w \in \smallset{a,b,c}^*}{\abs{w} \equiv 0 \bmod
      3}$. We know by \cite{DiekertKRW15jacm} that $L_7$ is 
      strongly \CR congruential, because it is regular, but we failed 
      to construct the corresponding \CR system within a reasonable 
      amount of time. \eex  
      \end{itemize}
The increasing difficulty in finding \CR systems for the languages 
$L_5$, $L_6$, $L_7$ is 
 in the increasing number of letters, and  to cope with commutative structures. Somehow 
    commutativity plays against \CR systems. 

The solution of \refcon{con:MNO88} is a typical example for the principal of \emph{loading induction}: proving a more statement 
is sometimes easier because a stronger inductive assumptions can be used.
\refcon{con:MNO88} speaks about \CR languages. First we replace the 
statement by an ``\IFF'' condition. We will show that 
a language is regular \IFF it is strongly \CR congruential.
Next, and this is crucial, we replace the existence of a 
finite \CR system by starting with an arbitrary weighted alphabet 
$(A,\Abs{\cdot})$ and we consider omly
finite confluent system $S\sse A^*\times A^*$ of finite index 
which are \wred for the given weight. 
Finally we switch to a purely  
algebraic statement.  Thus, 
instead of proving 
\refcon{con:MNO88} we actually proved the following  result
about \homo{}s to finite monoids.
\begin{theorem}\label{thm:main}
  Let $(A,\Abs{\cdot})$ be a weighted alphabet and let $\varphi : A^*
  \to M$ be a homomorphism to a finite monoid $M$. Then there exists a
  finite  confluent system $S\sse A^*\times A^*$ of finite index such that $\Abs \ell > \Abs r$ for all $\ell \to r \in S$ and such that 
  $\varphi$ factorizes through $S$.
\end{theorem}

\begin{corollary}\label{cor:main}
  A language $L \subseteq A^*$ is regular if and only if there exists
  a Church-Rosser system $S$ of finite index such that $L =
  \bigcup_{u \in L} [u]_S$. In particular, 
  all regular languages are strongly
Church-Rosser congruential.
\end{corollary}

The proof of \refthm{thm:main} is split in two main 
parts. First, the case where $\phi(A^*)$ is a group has to be solved. 
Below we show the solution for finite 
non-cyclic simple groups, the general case where $G$ is a group is much harder, and we do not go into details here, but refer to \cite{DiekertKRW15jacm}.
In Proposition~\ref{prp:main} we show how  the local divisor technique enables us to lift the result for groups to arbitrary monoids. These two parts are taken almost verbatim from \cite{DiekertKRW15jacm}.
We use the following convention. 
Let $(A,\Abs{\cdot})$ be a weighted alphabet and 
$S\sse A^*\times A^*$  be a finite confluent semi-Thue system 
of finite index. We then say that $S$ is a  \emph{weighted Church-Rosser
system} if in addition $\Abs \ell > \Abs r$ for all $\ell \to r \in S$.

The existence of weighted Church-Rosser
systems is easy to show for finite (non-cyclic) simple groups. Let $\phi: A^* \to G$ be a homomorphism to a finite
group, where $(A,\Abs{\cdot})$ is a weighted alphabet and  $L_G = \set{w \in A^*}{\phi(w) = 1}$. Let us assume
that the greatest common divisor $\gcd{\set{\Abs{w}}{w \in L_G}}$ is
equal to one; e.g.{} $\oneset{6,10,15} \sse \set{\Abs{w}}{w \in
  L_G}$. Then there are two words $u,v \in L_G$ such that $\Abs{u}
-\Abs{v} = 1$. Now we can use these words to find a constant $d$ such
that all $g\in G$ have a representing word $v_g$ with the exact weight
$\Abs{v_g} = d$. To see this, start with some arbitrary set of
representing words $v_g$. We multiply words $v_g$ with smaller weight
with $u$ and words $v_g$ with larger weights with $v$ until all weights are
equal.

The final step is to define the following weight-reducing system
\begin{equation*}
  S_G= \set{w \to v_{\phi(w)}} {w \in A^{*}\;  \text{ and } 
    d < \Abs w \leq d + \max\set{\Abs a}{a \in A}}.
\end{equation*}
Confluence of $S_G$ is trivial; and every language recognized by
$\phi$ is also recognized by the canonical \homo $A^* \to A^* / S_G$.

Now assume that we are not so lucky, i.e., $\gcd{\set{\Abs{w}}{w \in
    L_G}} >1$.  This means there is a prime number $p$ such that $p$
divides $\Abs{w}$ for all $w \in L_G$. Then, the \homo of $A^*$ to
$\Z/p\Z$ defined by $a \mapsto \Abs{a} \bmod p$ factorizes through
$\phi$ and $\Z/p\Z$ becomes a quotient group of $G$.  This can never
happen if $\phi(A^*)$ is a simple and non-cyclic subgroup of $G$,
because a simple group does not have any proper quotient group. But
there are many other cases where a natural homomorphism $A^* \to G$
for some weighted alphabet $(A,\Abs{\cdot})$ satisfies the property
$\gcd{\set{\Abs{w}}{w \in L_G}= 1}$ although $G$ has a non-trivial
cyclic quotient group.  Just consider the length function and a
presentation by standard generators for dihedral groups $D_{2n}$ or
the permutation groups $\mathcal{S}_n$ where $n$ is odd.

For example, let $G = D_6 = \mathcal{S}_3$ be the permutation group of
a triangle.  Then~$G$ is generated by elements $\tau$ and $\rho$ with
defining relations
$\tau^2 = \rho^3 = 1 \text { and } \tau \rho \tau = \rho^2$.
The following six words of length $3$ represent all six group
elements: 
\begin{equation*}
  1=\rho^3,\; \rho=\rho\tau^2,\; \rho^2 = \tau \rho \tau,\; \tau=\tau^3,\;
  \tau\rho = \rho^2\tau,\;  \tau\rho^2.
%
\end{equation*}
The corresponding monoid $\smallset{\rho,\tau}^* / S_G$ has $15$
elements. More systematically, one could obtain a normal form of
length $5$ for each of the group elements in
$\oneset{1,\rho,\rho^2,\tau,\tau\rho,\tau\rho^2}$ by adding factors
$\rho^3$ and $\tau^2$. For example, this could lead to the set of
normal forms
$\oneset{\tau^2\rho^3,\tau^4\rho,\rho^5,\tau^5,\tau\rho^4,\tau^3\rho^2}$.

\begin{proposition}[\cite{DiekertKW12tcs},\cite{DiekertKRW15jacm}]\label{prp:main}
Let $(A,\Abs{\cdot})$ be a weighted alphabet and let $\varphi : A^*
  \to M$ be a homomorphism to a finite monoid $M$.
  Assume that for all weighted alphabets  $(B,\Abs{\cdot}_B)$ and for all $\psi : B^*
  \to M'$ where $M'$ is a divisor of $M$, there exists some
\CR system $T\sse B^*\times B^*$ such that
  $\psi$ factorizes through $T$ whenever either $\abs {B} < \abs A$ and 
  $M' = M$ or $\abs {M'} < \abs M$ (or both).  
  Then there exists 
   some \CR system $S\sse A^*\times A^*$ of finite index such that
  $\phi$ factorizes through $S$.
\end{proposition}

\begin{proof} 
If $\varphi(A^*)$ is a finite non-cyclic group, then the claim is shown above. For other  groups we refer to  
  \cite{DiekertKRW15jacm}. If $\varphi(A^*)$ is not a group, then there
  exists $c \in A$ such that $\varphi(c)$ is not a unit.  Let $B = A
  \setminus \smallset{c}$. 
  By hypothesis 
  there exists a weighted Church-Rosser system $R$ for the
  restriction $\varphi : B^* \to M$ satisfying the statement of the
  theorem. Let
  \begin{equation*}
    K = \IRR_R(B^*) c.
  \end{equation*}
  We consider the prefix code $K$ as a weighted alphabet. The weight
  of a letter $uc \in K$ is the weight $\Abs{uc}$ when read as a word
  over the weighted alphabet $(A,\Abs{\cdot})$. Let $M_c = \varphi(c)
  M \cap M \varphi(c)$ be the local divisor of $M$ at $\varphi(c)$.
  We let $\psi : K^* \to M_c$ be the homomorphism induced by $\psi(uc)
  = \varphi(cuc)$ for $uc \in K$. By hypothesis
  there exists a weighted Church-Rosser system $T
  \subseteq K^* \times K^*$ for $\psi$ satisfying the statement of the
  theorem. Suppose $\psi(\ell) = \psi(r)$ for $\ell,r \in K^*$ and let
  $\ell = u_1 c \cdots u_j c$ and $r = v_1 c \cdots v_k c$ with
  $u_i,v_i \in \IRR_R(B^*)$. Then
  \begin{align*}
    \varphi(c \ell) 
    &= \varphi(cu_1c) \circ \cdots \circ \varphi(cu_jc) \\
    &= \psi(u_1 c) \circ \cdots \circ \psi(u_j c) \\
    &= \psi(\ell) = \psi(r) = \varphi(c r).
  \end{align*}
  This means that every $T$-rule $\ell \to r$ yields a
  $\varphi$-invariant rule $c\ell \to cr$.  We can transform the
  system $T \subseteq K^* \times K^*$ for $\psi$ into a system $T'
  \subseteq A^* \times A^*$ for $\varphi$ by
  \begin{equation*}
    T' = \set{c\ell \to cr \in A^* \times A^*}{\ell \to r \in T}.
  \end{equation*}
  Since $T$ is confluent and weight-reducing over $K^*$, the system
  $T'$ is confluent and weight-reducing over $A^*$. Combining $R$ and
  $T'$ leads to
    $S = R \cup T'$.
  The left sides of a rule in $R$ and a rule in $T'$ cannot overlap.
  Therefore, $S$ is a weighted Church-Rosser system such that
  $\varphi$ factorizes through $A^* / S$. Suppose that every word in
  $\IRR_T(K^*)$ has length at most $k$. Here, the length is over the
  extended alphabet $K$. Similarly, let every word in $\IRR_R(B^*)$
  have length at most~$m$. Then
  \begin{equation*}
    \IRR_S(A^*) \subseteq \set{u_0 c u_1 \cdots c u_{k'+1}}{ 
      u_i \in \IRR_R(B^*), \; k' \leq k}
  \end{equation*}
  and every word in $\IRR_S(A^*)$ has length at most $(k+2)m$. In
  particular $\IRR_S(A^*)$ and $A^* / S$ are finite.
\end{proof}

By Proposition~\ref{prp:main} one can easily see that it works verbatim when we replace ``\CR system'' by
``finite \swr confluent system of finite index''. Let us call such a  system a \emph{\swr \CR system}. We then can state a sharper  result than \refthm{thm:main} for aperiodic monoids. 

\begin{theorem}[\cite{DiekertKW12tcs}]\label{thm:mainap}
  Let $\varphi : A^*
  \to M$ be a homomorphism to a finite aperiodic monoid $M$. Then there exists a
  \swr \CR system $S\sse A^*\times A^*$  such that 
  $\varphi$ factorizes through $S$.
\end{theorem}

\section{Factorization forests}\label{sec:forests}
Factorization forests where introduce by Imre Simon in \cite{sim90}. 
The main result about factorization forests is a ``nested variant'' of Ramsey's Theorem. Given a regular language it associates to every word 
in a factorization tree of constant height (which depends on the syntactic monoid of the language). This deep insight of Simon has numerous applications, see for example the recent handbook-survey of Colcombet in \cite{Colcombet2021HandbookAuto}. Here, we give a simple proof using local divisors that factorization forests exist.

In the following let $M$ be a finite monoid and $A$ be a finite
alphabet.  A \emph{factorization forest} of a homomorphism $\varphi :
A^* \to M$ is a function $d$ which maps every word $w$ with
length $\abs{w} \geq 2$ to a factorization $d(w) = (w_1, \ldots,
w_{n})$ of $w = w_1 \cdots w_{n}$ such that $n \geq 2$ and
either $n=2$ or $\varphi(w_1) = \cdots = \varphi(w_{n})$ is idempotent
in $M$. An element $e \in M$ is idempotent, if $e^2 = e$.

Usually we avoid empty words in the factorization and then,
with respect to $d$,  every non-empty word can be visualized as a tree
where the leaves are labeled with letters. 
Thus, $d$ defines a \emph{factorization tree} for each word $w$.

The \emph{height} $h$ of a word $w$ is defined as
\begin{equation*}
  h(w) = 
  \begin{cases}
    0 & \text{if } \abs{w} \leq 1 \\
    1 + \max\oneset{h(w_1), \ldots, h(w_{n})} & 
    \text{if } d(w) = (w_1, \ldots, w_{n})
  \end{cases}
\end{equation*}

Let us say that $d$ is \emph{optimal} if the for each word $w$ with
length $\abs{w} \geq 2$ the height $d(w)$ is minimal.  A famous
theorem of Simon says that every homomorphism $\varphi : A^* \to
M$ has an optimal factorization forest of height $\Oh(\abs{M})$. The original proof of Simon was rather technical. A
simplified proof with a worse bound based on the Krohn-Rhodes decomposition
was found by Simon in \cite{sim92}.  Later improved bounds were
found using Green's relations \cite{ChaLeu04,kuf08mfcs}. 
However, in many cases it is enough to know that there is a
factorization forest of bounded height but the actual bound is not
used.  Moreover, the hard part in the proof usually is when the
underlying monoid is aperiodic.  So, it might be a good idea to base a
proof upon local divisors. As we see now, it works:

\begin{theorem}[Simon]\label{thm:factfor:ap}
  Let $M$ be a finite 
  monoid.  There is a constant $h(\abs M)$ such that every
  homomorphism $\varphi : A^* \to M$ has a factorization forest
  of height at most $h(\abs M)$.
\end{theorem}

\begin{proof}
  Let $\varphi : A^* \to M$ be a homomorphism. The first thing we
  observe is that we may assume $\abs{A} \leq \abs{M}$, because
  if different letters are mapped to the same element in $M$ we can
  identify these letters without changing the height. Thus, actually
  we may assume that $A$ is a subset of $M$ and $\varphi$ is
  induced by this inclusion.

  The case where $M$ is a finite group $G$ is rather simple and nicely
  exposed in \cite{ChaLeu04}. For convenience we repeat the argument.
  Let $g_1 \cdots g_n$ be a (long) word of group elements.  The basic
  idea is to perform an induction on the size of the \emph{prefix set}
  which is defined by
  \begin{equation*}
    P(g_1 \cdots g_n) =\set{g_1 \cdots g_i \in G}{1\leq i \leq  n} 
  \end{equation*}
  Choose some maximal subset $\oneset{i_1,\ldots, i_t} $ of
  $\oneset{1,\ldots, n}$ 
  such that all prefixes $g_1 \cdots g_{i_j}\in G$ are equal.  We may
  assume that $t\geq 2$.
  Let $i_0 = 0$ and $i_{t+1} =n$. Consider the $t+1$ factors $v_j =
  g_{i_{j-1}+1} \cdots g_{i_j}$. Thus, the word $g_1 \cdots g_n$
  facorizes as $v_1 \cdots v_{t+1}$. Let $h_j \in G$ be the evaluation
  of the word $v_j$ in $G$ for $1 \leq j \leq t+1$.  Then we have $h_2
  = \cdots = h_t =1$ and we are done if the size of each prefix set
  $P(v_j) $ is striclty less than the size of $P(g_1 \cdots g_n)$.
  It is clear that the size of the prefix set of $g_1 \cdots g_{i_1}$
  has decreased, but this is actually the case for all $g_{i_{j-1}+1}
  \cdots g_{i_j}$. Indeed for $2\leq j \leq t+1$ we have
  \begin{equation*}
    h_1\cdot P(g_{i_{j-1}+1} \cdots g_{i_j}) \subseteq 
    P(g_1 \cdots g_n)\setminus \smallset{h_1}
  \end{equation*}
  The result for $G$ follows because the translation by any group
  element is injective.

  Now let $M$ be an arbrary finite monoid.  Consider a word $w$ where
  the letters are elements of $M$. If all letters are units, then $w$
  is mapped to a subgroup $G$ of $M$. We are done by the case above.
  Therefore we may assume that in $w$ some letter $c$ occurs which is
  not a unit.  In particular $1 \not\in cM \cap Mc$. Then $w$ admits a
  factorization
  \begin{equation*}
    w = w_0 c w_1 c w_2 \cdots c w_k c w_{k+1}
  \end{equation*}
  where $c$ does not occur in any $w_i$ for $0 \leq i \leq k+1$. By
  induction on the alphabet size of $w$, there exist factorization
  trees of small height for all $w_i$, and thus, for each factor
  $cw_i$ as well. This allows us to treat factors $cw_i$ as
  letters. More precisely, let $w' = c w_1 c w_2 \cdots c w_k$. It is
  clear that we may assume $w=w'$. We read each factor $cw_i$ as a new
  letter $b_i$ in some alphabet $T$ and we read $w'$ as a word
  $b_1 \cdots b_k$ in $T^*$.  Consider the homomorphism $\psi:
  T^* \to M_c$ induced by $\psi(cw_i) = cw_ic$, where $M_c = cM
  \cap M c$ is the local divisor of $M$ at $c$.  By induction on the
  size of the monoid, there exists a factorization forest $d_c$ for
  the homomorphism $\psi$ of height $h(\abs{M} -1)$.  Inductively we
  transform the factorization tree of $b_1 \cdots b_k$ into a
  factorization tree of $w'$.
  If $d_c(b_1 \cdots b_k) = (b_1 \cdots b_i, \, b_{i+1} \cdots b_k)$
  then we let
  \begin{equation*}
    d(w') =  (c_1w_1 \cdots cw_i, \, cw_{i+1} \cdots cw_k)
  \end{equation*}
  We now treat the case $d_c(b_1 \cdots b_{k})
  = (v_1, v_2, \ldots, v_{\ell})$ with $\ell \geq 3$.
  Each $v_j$ corresponds to an element in $M_c^*$ of the form $(c
  w_{i_j} c) \cdots (c w_{i_{j+1}-1} c)$, and we let $\widetilde{w}_j
  = w_{i_j} c \cdots c w_{i_{j+1}-1} \in M$.
  We choose some maximal subset $\smallset{i_1,\ldots, i_t} $ of
  $\smallset{1,\ldots, \ell} $ such that we have both,
  $c\widetilde{w}_{i_1} = \cdots = c\widetilde{w}_{i_t}$ in $M$ and a
  factorization of the form
    \begin{equation*}
    c \widetilde{w}_1 c \widetilde{w}_2 \cdots c \widetilde{w}_{\ell}
    =
    c \overline{w}_1 c \widetilde{w}_{i_1} 
    c \overline{w}_2 c \widetilde{w}_{i_2} \cdots 
    c \overline{w}_t c \widetilde{w}_{i_t}
    c \overline{w}_{t+1}
  \end{equation*}
  for some $\overline{w}_i \in M^*$.  A straightforward computation
  shows that elements $c \overline{w}_j c \widetilde{w}_{i_j}$ are all
  identical and idempotent in $M$ since the $v_i$ are idempotent in
  $M_c$. This yields the factorizations
  \begin{align*}
    d(c \widetilde{w}_1 
    \cdots c \widetilde{w}_{\ell})
    &= (c \overline{w}_1 c \widetilde{w}_{i_1} 
    \cdots 
    c \overline{w}_t c \widetilde{w}_{i_t},\, c \overline{w}_{t+1}) \\
    d(c \overline{w}_1 c \widetilde{w}_{i_1} 
    \cdots 
    c \overline{w}_t c \widetilde{w}_{i_t})
    &= (c \overline{w}_1 c \widetilde{w}_{i_1},\, 
    c \overline{w}_2 c \widetilde{w}_{i_2},\, 
    \ldots, \, 
    c \overline{w}_t c \widetilde{w}_{i_t})
  \end{align*}
  The next step splits each $c \overline{w}_j c \widetilde{w}_{i_j}$
  in two factors using $d(c \overline{w}_j c \widetilde{w}_{i_j}) =(c
  \overline{w}_j,\, c \widetilde{w}_{i_j}) $.  The factor $c
  \widetilde{w}_{i_j}$ corresponds to $v_j$, and either $c
  \overline{w}_j$ corresponds to some $v_{j'}$ or the element $c
  \widetilde{w}_{i_j}$ does not occur in the factor $c
  \overline{w}_j$. Thus in at most $\abs{M}$ transformation steps each
  remaining factor is of the form $w_i$ and we are done, since the
  letter $c$ did vanish. 
\end{proof}

{\small
\newcommand{\Th}{Th}\newcommand{\Ch}{Ch}\newcommand{\Yu}{Yu}\newcommand{\Zh}{Zh}\newcommand{\St}{St}\newcommand{\curlybraces}[1]{\{#1\}}

}

\end{document}